\documentclass[12pt]{article}
\usepackage{amsfonts}
\usepackage{amssymb}
\usepackage{color}
\usepackage{graphicx}

\newcommand{\F}{{\cal F}}
\newcommand{\Sc}{{\cal S}}

\newcommand{\mc}{\mathcal}

\newcommand{\be}{\begin{equation}}
\newcommand{\en}{\end{equation}}
\newcommand{\bea}{\begin{eqnarray}}
\newcommand{\ena}{\end{eqnarray}}
\newcommand{\beano}{\begin{eqnarray*}}
\newcommand{\enano}{\end{eqnarray*}}
\newcommand{\D}{{\mc D}}

\newcommand{\G}{{\cal G}}

\newcommand{\ltwo}{{\Lc^2(\mathbb{R})}}
\newcommand{\ltwod}{{\Lc^2(\mathbb{R}^d)}}
\newcommand{\ltwodue}{{\Lc^2(\mathbb{R}^2)}}
\newcommand{\scrdue}{{\Sc(\mathbb{R}^2)}}

\renewcommand{\l}{\langle}
\renewcommand{\r}{\rangle}
\newcommand{\pin}[2]{\l#1 , #2\r}

\newcommand{\Lc}{{\cal L}}
\newcommand{\1}{1 \!\! 1}

\newcommand{\Hil}{\mc H}

\newcommand{\un}{{\underline n}}
\newcommand{\um}{{\underline m}}
\newcommand{\ux}{{\underline x}}
\newcommand{\uzero}{{\underline 0}}
\newcommand{\uuno}{{\underline 1}}
\newcommand{\uk}{{\underline k}}
\newcommand{\ul}{{\underline l}}

\newtheorem{thm}{Theorem}

\newtheorem{prop}[thm]{Proposition}
\newtheorem{defn}[thm]{Definition}

\newenvironment{proof}{\noindent {\bf Proof:}}{\hfill$\Box$}

\catcode `\@=11 \@addtoreset{equation}{section}
\def\theequation{\arabic{section}.\arabic{equation}}
\catcode `\@=12
\begin{document}

\begin{center}
{\Large \textbf{Multiplication of distributions in a linear gain and loss system}} \vspace{2cm%
}\\[0pt]

{\large F. Bagarello}
%\footnote[1]{ Dipartimento di Matematica ed Applicazioni,
%Fac.Ingegneria, Universit\`a di Palermo, I - 90128  Palermo, Italy}
\vspace{3mm}\\[0pt]
Dipartimento di Ingegneria,\\[0pt]
Universit\`{a} di Palermo, I - 90128 Palermo,\\
and I.N.F.N., Sezione di Catania\\
E-mail: fabio.bagarello@unipa.it\\

\vspace{7mm}

\end{center}

\vspace*{2cm}

\begin{abstract}
\noindent We consider a model of coupled oscillators which can be seen as a gain and loss system. In the attempt to quantize the system we propose a new definition of multiplication between distributions, and we check that this definition can be adopted when checking the biorthonormality of the eigenstates of the Hamiltonian $H$ of the system, and of its adjoint $H^\dagger$. In the analysis carried out here, the role of {\em weak pseudo-bosonic} ladder operators is relevant.  
\end{abstract}

\vspace*{1cm}


\vfill

\newpage

% Section 1

\section{Introduction}

Since many decades simple electronic circuits have been used in connection with classical and, more recently, quantum mechanics, since they can provide concrete devices where some physical effect can be observed, or modelled. A well known example of such a map between Electronics and Classical Mechanics is given by  any {\em RLC-circuit} (RLCc), which is dynamically equivalent to a {\em damped harmonic oscillator} (DHO). This is because the time evolution of the charge in an RLCc, i.e. a circuit with an inductance, a resistance and a capacitor in series, is driven by exactly the same equation of motion of that of a DHO, with the dissipative effect of the resistance replaced by the friction, see e.g. \cite{wells,panu} and references therein. In \cite{baldiotti}-\cite{review2}, among others, a quantum version of the DHO, and therefore of the RLCc, has been considered both for a purely mathematical interest and in view of possible applications. In particular, due to its particularly simple form, a deep comprehension of the DHO/RLCc is surely a first step toward a better understanding of the correct quantization procedure for generic dissipative systems, see \cite{kanai,vitiello}, which are so relevant in several physical contexts.

 A similar interest is at the basis of other papers connected with related problems, \cite{circu2}-\cite{circu5}: in all these papers electronic circuits are analyzed, either in connection with PT-Quantum Mechanics and exceptional points, \cite{benbook}, or because they produce interesting results when one tries to quantize the system, not only for the physical consequences of this quantization, but also because of its many mathematical consequences.
 
 For instance, some interesting mathematics appears when  quantizing a DHO, and diagonalizing its Bateman Hamiltonian, \cite{bate}, by means of ladder operators. In particular, in \cite{nakano,deguchi} the authors claimed they could construct two biorthonormal bases of square integrable eigenfunctions for the Hamiltonian of the DHO. However, as shown in \cite{BGR,BGR2}, this claim was wrong, since in particular the vacua of the lowering operators needed in their construction are not square integrable functions, but Dirac delta distributions. Hence distributions are more relevant than functions in this context. And, in fact,  this was not the first appearance of distributions, and of the Dirac delta in particular, in the analysis of some physical system. For instance, they are the main object of research when dealing rigorously with quantum fields, \cite{streit}. More recently, distributions have proved to be the natural tool to use when looking for the (generalized) eigenstates of certain Hamiltonians, mostly connected with what we have called {\em weak pseudo-bosons}, see e.g. \cite{bag2020JPA}-\cite{bagJPA2022} and \cite{bagspringer} for a recent monograph on this and related topics.

 { In this paper we show how an extended version of the Bateman Lagrangian can be introduced and studied, and the role that distributions play in this analysis. More in details, we start  replacing the original pairs of uncoupled oscillators described by Bateman with different pairs of possibly interacting oscillators, one of which is still damped (hence, is a loss system), while the other is {\em amplified} (so that it can be interpreted as a gain system). The Bateman system can be recovered as a special case of our settings. The main physical result of our analysis is that we will able to quantize the system, and to find the eigenvalues and the eigenvectors of its related Hamiltonian $H$. However, these eigenvectors turn out {\bf not} to be functions. Indeed, they are distributions and, as such, they require some extra mathematical care to be properly considered, also in view of their relations with the eigenvectors of $H^\dagger$, the adjoint of $H$. In particular, biorthormality of these two families of eigenvectors\footnote{We should recall that, in presence of non self-adjoint operators, orthonormality of eigenstates is usually replaced by biorthogonality of the eigestates of the {\em physically relevant} operator (e.g., the Hamiltonian $H$), and those of its adjoint (e.g., $H^\dagger$), \cite{benbook,bagspringer}} should be defined, since no {\em ordinary} scalar product can be naturally introduced in our analysis. This aspect is discussed in many concrete situations in the literature, \cite{benbook,bagspringer}, but mostly in a Hilbert space settings. Here, Hilbert spaces are not enough. For this reason, to analyze biorthonormality of these vectors, we need to consider a new concept of scalar product, which extends the usual one in $\ltwo$. This is possibly the most interesting (mathematical) result of this paper. This definition will be applied to our specific system, and indeed a kind of biorthonormality will be deduced. Surprisingly enough, the notion of Abel summation will be quite relevant in this analysis.

% {\color{red}In this paper we consider an alternative interpretation of the Bateman Lagrangian, by considering two possibly {\bf concrete} coupled systems, a loss system (the DHO), interacting with a gain system (an {\em amplified harmonic oscillator} AHO). This is different from the Bateman's approach, where to describe a single real DHO a second {\em virtual} amplified oscillator is also introduced which behaves as a gain system, earning the energy which is dissipated by the DHO. Here, rather than looking the second as a virtual system, we consider it like a real system, possibly implemented by a specific RLC electronic circuit coupled to a second circuit where the resistor $R$ is replaced by an amplifier which models a negative resistor. Moreover, with respect to the Bateman's approach, the two oscillators we analyze are coupled.
%  This is the system we will consider and quantize in this paper, again in terms of ladder operators of the pseudo-bosonic type. Dealing with it, we will be forced to understand if and how some distributions can be multiplied. This will be possible adopting a new definition of multiplication between distributions, which appears more friendly than the one we have already considered in previous papers, \cite{BGR}-\cite{bagspringer}, and originally proposed in some mathematical literature on distributions, \cite{morton}-\cite{vlad}. This new definition is, in our opinion, one of the main contributions of this work.}
 
The paper is organized as follows: in Section \ref{sect2} we will review some previous results on the DHO, to introduce the notation and to stress some essential aspects of what was already discussed in the literature. Also, this preliminary analysis is useful since the extended system considered in this paper reduces to that, after a suitable change of variables. In Section \ref{sect3} we introduce our gain-loss linear circuit, and we show how to introduce ladder operators in its description. As already stated, this forces us to deal with distributions. This motivates our analysis in Section \ref{sect4}, where we propose a new definition of multiplication between distributions, and we deduce some of its properties. In Section \ref{sect5} we then show how this multiplication can be used in the analysis of our specific gain-loss system. Section \ref{sect6} contains our conclusions.{ To keep the paper self-contained, we  list a series of definitions and properties of pseudo-bosonic ladder operators in Appendix A, while in Appendix B we prove some useful identities used in Section \ref{sect5}. Also,  Appendix C contains some technical results useful for us.}

}

\section{Preliminaries}\label{sect2}

We devote this section to a brief review of what was discussed in \cite{BGR,BGR2}. This is relevant in view of what follows, since we will show that the system introduced in Section \ref{sect3} can be rewritten as the one we are going to consider here. 

The classical equation for the DHO is $m\ddot x+\gamma \dot x+kx=0$, in which $m,\gamma$ and $k$ are the physical positive quantities of the oscillator: the mass, the friction coefficient and the spring constant\footnote{They are in one-to-one correspondence with the inductance $L$, the resistance $R$ and the inverse capacity $C^{-1}$ of a RLCc.}. The Bateman lagrangian, \cite{bate}, is
\be L_0=m\dot x\dot y+\frac{\gamma}{2}(x\dot y-\dot xy)-kxy,
\label{21}\en
which, other than the previous equation, produces also $m\ddot y-\gamma \dot y+ky=0$, the differential equation associated to the virtual AHO, see also \cite{fesh}. The conjugate momenta are $$p_x=\frac{\partial L_0}{\partial \dot x}=m\dot y-\frac{\gamma}{2}\,y,\qquad
p_y=\frac{\partial L_0}{\partial \dot y}=m\dot x+\frac{\gamma}{2}\,y,
$$
and the corresponding classical Hamiltonian is
\be
H_0=p_x\dot x+p_y \dot y-L_0=\frac{1}{m} p_xp_y+\frac{\gamma}{2m}\left(yp_y-xp_x\right)+\left(k-\frac{\gamma^2}{4m}\right)xy.
\label{22}\en
By introducing the new variables $x_1$ and $x_2$ through 
\be
x=\frac{1}{\sqrt{2}}(x_1+x_2), \qquad y=\frac{1}{\sqrt{2}}(x_1-x_2),
\label{23}\en
$L_0$ and $H_0$ can be written as follows:
$$
L_0=\frac{m}{2}(\dot x_1^2-\dot x_2^2)+\frac{\gamma}{2}(x_2\dot x_1-x_1\dot x_2)-\frac{k}{2}(x_1^2-x_2^2)
$$
and
$$
H_0=\frac{1}{2m}\left(p_1-\frac{\gamma}{2}x_2\right)^2-\frac{1}{2m}\left(p_2-\frac{\gamma}{2}x_1\right)^2+\frac{k}{2}(x_1^2-x_2^2),
$$
where $p_1=\frac{\partial L_0}{\partial \dot x_1}=m\dot x_1+\frac{\gamma}{2}\,x_2$ and $p_2=\frac{\partial L_0}{\partial \dot x_2}=m\dot x_2-\frac{\gamma}{2}\,x_1$. By putting $\omega^2=\frac{k}{m}\,-\frac{\gamma^2}{4m^2}$ we can rewrite $H_0$ as follows:
\be
H_0=\left(\frac{1}{2m}p_1^2+\frac{1}{2}m\omega^2x_1^2\right)-\left(\frac{1}{2m}p_2^2+\frac{1}{2}m\omega^2x_2^2\right)-\frac{\gamma}{2m}(p_1x_2+p_2x_1).
\label{24}\en
We will here only consider $\omega^2>0$. The case {$\omega^2\leq0$}  has been briefly considered in \cite{BGR}.

Following \cite{nakano} we impose the following canonical quantization rules between $x_j$ and $p_k$: $[x_j,p_k]=i\delta_{j,k}\1$, working in unit $\hbar=1$. Here $\1$ is the identity operator. This is equivalent to the choice in \cite{fesh}. Ladder operators can now be easily introduced:
\be a_k=\sqrt{\frac{m\omega}{2}}\,x_k+i\sqrt{\frac{1}{2m\omega}}\,p_k,
\label{25}\en
$k=1,2$. These are bosonic operators since they satisfy the canonical commutation rules: $[a_j,a^\dagger_k]=\delta_{j,k}\1$. Furthermore, they are densely defined on any Schwartz test function. In particular, $[a_j,a^\dagger_k]\varphi(x)=\delta_{j,k}\varphi(x)$, for all $\varphi(x)\in\Sc(\mathbb{R})$. It might be useful to recall that $\Sc(\mathbb{R})$ is the set of all $C^\infty$ functions which decrease, together with their derivatives, faster than any inverse power of $x$, \cite{vlad}:
$$\Sc({\mathbb{R}})= \left\{g(x)\in C^\infty: \lim_{|x|,\,\infty}
{|x|^k g^{(l)} (x)}=0 \quad \forall k,l \in {\mathbb{N}_0}\right\},$$
where $\mathbb{N}_0=\mathbb{N}\cup\{0\}$.
\vspace{2mm}

In terms of these operators the quantum version of the Hamiltonian $H_0$ in (\ref{24}) can be written as
\be
	H_0=\omega\left(a_1^\dagger a_1-a_2^\dagger a_2\right)+\frac{i\gamma}{2m}\left(a_1a_2-a_1^\dagger a_2^\dagger\right).
\label{26}\en

Following again \cite{nakano} we introduce the operators:
\be
A_1=\frac{1}{\sqrt{2}}(a_1-a_2^\dagger), \quad A_2=\frac{1}{\sqrt{2}}(-a_1^\dagger+a_2),
\label{27}\en
as well as 
\be
B_1=\frac{1}{\sqrt{2}}(a_1^\dagger+a_2), \quad B_2=\frac{1}{\sqrt{2}}(a_1+a_2^\dagger).
\label{28}\en
They satisfy the following requirements:
\be
[A_j,B_k]\varphi(x)=\delta_{j,k}\varphi(x),
\label{29}
\en
$\forall \varphi(x)\in\Sc(\mathbb{R})$. We observe that $B_j\neq A_j^\dagger$, $j=1,2$. Moreover, $A_1=-A_2^\dagger$ and $B_1=B_2^\dagger$. It might be useful to stress that the map in (\ref{27})-(\ref{28}) is reversible, since $a_j$ and $a_j^\dagger$ can be recovered out of $A_j$ and $B_j$. 

In \cite{bagspringer,baginbagbook} operators of this kind, named {\em pseudo-bosonic}, were analyzed in detail, producing several interesting results mainly connected with their nature of ladder operators.

In terms of these operators $H_0$ can now be written as follows:
\be
	H_0=\omega\left(B_1A_1-B_2A_2\right)+\frac{i\gamma}{2m}\left(B_1A_1+B_2A_2+\1\right),\\
\label{210}\en 
which only depends on the pseudo-bosonic number operators $N_j=B_jA_j$, \cite{baginbagbook}. This is exactly the same Hamiltonian found in \cite{nakano}, and it is equivalent to that given in \cite{dekk2,fesh} and in many other papers on this subject. In \cite{BGR} we proved the following theorem, stating that the pseudo-bosonic lowering operators $A_1$, $A_2$, $B_1^\dagger$ and $B_2^\dagger$ do not admit square integrable vacua.

\begin{prop}\label{prop1}
	There is no non-zero function $\varphi_{00}(x_1,x_2)$ satisfying $$A_1\varphi_{00}(x_1,x_2)=A_2\varphi_{00}(x_1,x_2)=0.$$ Also, there is no non-zero function $\psi_{00}(x_1,x_2)$ satisfying $$B_1^\dagger\psi_{00}(x_1,x_2)=B_2^\dagger\psi_{00}(x_1,x_2)=0.$$
\end{prop}

\vspace{2mm}

We refer to \cite{BGR,BGR2} for further results on this problems. In particular, it was shown that 
the vacua of $A_j$ and $B_j^\dagger$, $j=1,2$, are respectively $\varphi_{00}(x_1,x_2)=\alpha\delta(x_1-x_2)$ and $\psi_{00}(x_1,x_2)=\beta\delta(x_1+x_2)$: they are not functions, but distributions. $\alpha$ and $\beta$ are some sort of {\em normalization constants}.
 Here we just want to stress that, in these latter papers, our analysis stopped at this level because our interest was mainly focused in proving that it was not possible to find square-integrable eigenfunctions of $H_0$, contrarily to what claimed in \cite{nakano,deguchi}. What is more interesting for us, now, is the possibility to answer the following questions:

$\bullet$ is it possible to replace the pair DHO-AHO with some more general system sharing similar properties, and a similar quantization procedure?

$\bullet$ is it possible to construct a set of {\em biorthonormal-like} distributions out of $\varphi_{00}(x_1,x_2)$ and $\psi_{00}(x_1,x_2)$ using the pseudo-bosonic raising operators as described in Appendix A?

We will see that both these questions can be successfully considered, and in particular the second question will force us to introduce an interesting mathematical tool, which can be used to define a class of multiplications between distributions.

\section{Our system}\label{sect3}

In this section we will consider the first question raised above, proposing a simple classical Lagrangian which generalizes the one in (\ref{21}), and which describes, in view of its interpretation as a gain-loss system,  two coupled DHO-AHO. The idea is very simple: we just add to $L_0$ in (\ref{21}) another term, $L_1$, which is again quadratic in the variables $x$ and $y$.  With a proper choice of $L_1$, and of the parameters of the system, we will be able to describe a new pair of coupled oscillators, and to quantize the system in a similar way as we discussed in Section \ref{sect2}, facing with similar problems.

Let us consider
\be
L_1=A(m\dot x^2-kx^2)+B(m\dot y^2-ky^2)
\label{31}\en
and
\be
L=L_0+L_1=m\dot x\dot y+\frac{\gamma}{2}(x\dot y-\dot xy)-kxy+A(m\dot x^2-kx^2)+B(m\dot y^2-ky^2).
\label{32}\en
Here $A$ and $B$ are constants whose values will be constrained later.

\vspace{2mm}

{\bf Remark:--} The Lagrangian in (\ref{32}) is a particular case of a more general choice $L=L_0+L_1$, with $L_1=f(y,\dot y)+g(x,\dot x)$. Not surprisingly, also in this general case, if we write the Hamiltonian $H=p_x\dot x+p_y\dot y-L$, and we compute its time derivative, we get $\dot H=0$, which can be interpreted as some kind of energy conservation for the coupled system.

\vspace{2mm}

From $L$ in (\ref{32}) we get the following set of coupled differential equations:
\be\label{33}
\left\{
\begin{array}{ll}
	m\ddot x+\gamma \dot x+kx=-2B\left(m\ddot y+ky\right)\\
	m\ddot y-\gamma \dot y+ky=-2A\left(m\ddot x+kx\right),\\
\end{array}
\right.
\en
which can also be rewritten, after some minor manipulations, as
\be\label{34}
\left\{
\begin{array}{ll}
	m'\ddot x+\gamma \dot x+k'x=-2B\gamma \dot y\\
	m'\ddot y-\gamma \dot y+k'y=2A\gamma \dot x,\\
\end{array}
\right.
\en
where $m'=m(1-4AB)$ and $k'=k(1-4AB)$, which are both positive if $AB<\frac{1}{4}$. It is now possible to rewrite $L$ in a form which is quite close to $L_0$ in (\ref{21}), with a change of variable $(x,y)\rightarrow(X_1,Y_1)$:
\be\label{35}
\left\{
\begin{array}{ll}
	x=\alpha_x X_1+\beta_x Y_1,\\
	y=\alpha_y X_1+\beta_y Y_1,\\
\end{array}
\right.
\en
where $\alpha_x$, $\alpha_y$, $\beta_x$ and $\beta_y$ must satisfy the condition $\alpha_x\beta_y-\beta_x\alpha_y\neq0$, in order to have an invertible transformation. From now on, we take
\be
\alpha_x=-\,\frac{\alpha_y}{2A}\left(1-\sqrt{1-4AB}\right), \qquad \beta_x=-\,\frac{\beta_y}{2A}\left(1+\sqrt{1-4AB}\right),
\label{36}\en
so that $\alpha_x\beta_y-\beta_x\alpha_y=\frac{\alpha_y\beta_y}{A}\sqrt{1-4AB}$, which is different from zero if $\alpha_y,\beta_y\neq0$, under our constraint on $AB$. After some manipulation we get that
\be
L= m_1\dot X_1\dot Y_1+\frac{\gamma_1}{2}(X_1\dot Y_1-\dot X_1Y_1)-k_1X_1Y_1,
\label{37}\en
where we have introduced
\be
m_1=\frac{m\alpha_y\beta_y}{A}(4AB-1), \quad k_1=\frac{k\alpha_y\beta_y}{A}(4AB-1),\quad  \gamma_1=\frac{\gamma\alpha_y\beta_y}{A}\sqrt{1-4AB}.
\label{38}\en
Recalling that $4AB-1<0$, it is clear that $k_1,m_1>0$ only if $\frac{\alpha_y\beta_y}{A}<0$, which is what we will assume from now on. However, under this condition, it follows that $\gamma_1=-|\gamma_1|<0$. We rewrite $L$ in (\ref{37}) as
\be
L= m_1\dot X_1\dot Y_1+\frac{|\gamma_1|}{2}(Y_1\dot X_1-\dot Y_1X_1)-k_1X_1Y_1.
\label{39}\en
This Lagrangian describes again a coupled DHO-AHO as the original one in (\ref{21}), where $Y_1$ is the coordinate of the DHO ($Y_1\rightleftarrows x$), while $X_1$ is that of the AHO ($X_1\rightleftarrows y$). Hence we can repeat the same steps as in Section \ref{sect2}, and in particular quantize the system and diagonalize the Hamiltonian in terms of pseudo-bosonic operators.  Formula (\ref{23}) is replaced here by
$$
Y_1=\frac{1}{\sqrt{2}}(x_1+x_2), \qquad X_1=\frac{1}{\sqrt{2}}(x_1-x_2).
$$
Then, introducing $p_1$ and $p_2$ as before, $p_j=\frac{\partial L}{\partial \dot x_j}$, we get $p_1=m_1\dot x_1+\frac{|\gamma_1|}{2}\,x_2$ and $p_2=m_1\dot x_2-\frac{|\gamma_1|}{2}\,x_1$, and the classical Hamiltonian $H_0$ in (\ref{24}) should be replaced now by
\be
H=\left(\frac{1}{2m_1}p_1^2+\frac{1}{2}m_1\omega_1^2x_1^2\right)-\left(\frac{1}{2m_1}p_2^2+\frac{1}{2}m_1\omega_1^2x_2^2\right)-\frac{|\gamma_1|}{2m_1}(p_1x_2+p_2x_1),
\label{310}\en
where $\omega_1^2=\frac{k_1}{m_1}\,-\frac{|\gamma_1|^2}{4m_1^2}=\frac{k}{m}\,-\frac{\gamma^2}{4m^2(1-4AB)}$, which we assume here to be positive.

Next we quantize the system, requiring that
$[x_j,p_k]=i\delta_{j,k}\1$, and we introduce the bosonic operators
\be a_k=\sqrt{\frac{m\omega}{2}}\,x_k+i\sqrt{\frac{1}{2m\omega}}\,p_k,
\label{310bis}\en
$k=1,2$, and their combinations
\be
A_1=\frac{1}{\sqrt{2}}(a_1-a_2^\dagger), \quad A_2=\frac{1}{\sqrt{2}}(-a_1^\dagger+a_2),\quad B_1=\frac{1}{\sqrt{2}}(a_1^\dagger+a_2), \quad B_2=\frac{1}{\sqrt{2}}(a_1+a_2^\dagger).
\label{311}\en
These operators satisfy, see (\ref{29}), the commutation rule
\be
[A_j,B_k]\varphi(x)=\delta_{j,k}\varphi(x),
\label{312}
\en
$\forall \varphi(x)\in\Sc(\mathbb{R})$, as well as the other properties stated in Section \ref{sect2}. An essential consequence is that $H$ is diagonal in these operators,
\be
H=\omega_1\left(B_1A_1-B_2A_2\right)+\frac{i|\gamma_1|}{2m_1}\left(B_1A_1+B_2A_2+\1\right),\\
\label{313}\en 
 and Proposition \ref{prop1} applies. In particular, the vacua of $A_j$ and $B_j^\dagger$, $j=1,2$, are respectively $\varphi_{00}(x_1,x_2)=\alpha\delta(x_1-x_2)$ and $\psi_{00}(x_1,x_2)=\beta\delta(x_1+x_2)$.

Going back to our first question at the end of Section \ref{sect2}, we have seen here that it is indeed possible to construct more general systems\footnote{Of course $L$ in (\ref{32}) returns the original system in Section \ref{sect2} if $A=B=0$. Otherwise the two oscillators are coupled, which was not the case for $L_0$ in (\ref{21}).} which, after a certain change of variables, turn out not to be different from the pair of oscilators described by the Bateman Lagrangian. Next, because of the role the distributions play in our analysis, we consider a mathematical interlude on a possible definition of a class of multiplications between distributions. We should maybe stress that, in fact, the content of Section \ref{sect4} is (in our opinion) the most relevant mathematical result of this paper.

\section{Multiplication of distributions}\label{sect4}

In \cite{vlad} a possible way to introduce a multiplication between distributions was discussed. It is based on the simple fact that the scalar product between two {\em good } functions $f(x)$ and $g(x)$, for instance $f(x),g(x)\in\Sc(\mathbb{R})$, can be written in terms of a convolution between $\overline{f(x)}$ and $\tilde{g}(x)=g(-x)$: $\left<f,g\right>=(\overline{f}* \tilde{g})(0)$. Hence it is natural to define the scalar product between two elements $F(x), G(x)\in\Sc'(\mathbb{R})$ as the following convolution:
\be
\left<F,G\right>=(\overline{F}* \tilde{G})(0),
\label{41}\en
whenever this convolution exists, which is not always true. Notice that, in order to compute $\left<F,G\right>$, it is first necessary to compute $(\overline{F}* \tilde{G})[f]$, $f(x)\in\Sc(\mathbb{R})$, and this can be done by using the equality $(\overline{F}* \tilde{G})[f]=\left<F,G*f\right>$ which, again, is not always well defined. It is maybe useful to stress that $(\overline{F}* \tilde{G})[f]$ represents here the action of $(\overline{F}* \tilde{G})(x)$ on the function $f(x)$.

This approach has been used in some concrete situations in recent years, mainly to check if the generalized eigenstates of some non self-adjoint operator $\hat H$ are biorthonormal (with respect to this generalized product) to those of $\hat H^\dagger$. Some results in this direction can be found in \cite{bag2020JPA}-\cite{bagspringer}.

However, this approach does not seem to be flexible enough to cover also the situation discussed in this paper, i.e. to deal with the set of weak eigenvectors of the Hamiltonian of the system in Section \ref{sect3}. This is because, among other difficulties, it is very hard to take care properly of the domains of the various operators involved in this analysis, but also because it is quite complicated to perform explicit computations even in simple cases. For this reason we introduce now a different multiplication between distributions, and we analyze some of its properties. In Section \ref{sect5} we will show how this new definition works for our gain-loss system. To be general, we work here with $\ltwod$, $d\geq1$. First of all we introduce an orthonormal, total, set of vectors in $\ltwod$:
\be
\F_e=\{e_{\underline n}(\underline x)\in \Sc(\mathbb{R}^d), \quad{\underline n}=(n_1,n_2,\ldots,n_d) \},
\label{42}\en
where each $n_j=0,1,2,3,\ldots$. For instance, if $d=1$ the set $\F_e$ could be the set of the eigenstates of the quantum harmonic oscillator. If $d\geq2$, $\F_e$ can be constructed as tensor product of these 1-d functions, and so on.
Due to the nature of $\F_e$, for all $f(x), g(x)\in \ltwod$ we have that
$$
\pin{f}{g}=\sum_{\un}\pin{f}{e_\un}\pin{e_\un}{g}=\sum_{\un} \overline{f}[e_\un]\,g[\overline{e_\un}]=\sum_{\un} \overline{f}[e_\un]\,g[{e_\un}],
$$
if we further assume that each $e_\un(\underline x)$ is real, for simplicity. We are using the following notation:
\be
\overline{h}[c]=\int_{\mathbb{R}^d} \overline{h(\ux)}\,c(\ux)\, d\ux=\pin{h}{c}.
\label{42bis}\en
 In the Parceval identity above the particular choice of $\F_e$ is not relevant, as far as $f(x), g(x)\in \ltwod$. Now, the fact that $e_\un(\ux)\in\Sc(\mathbb{R}^d)$ implies that, for all $K(\ux)\in\Sc'(\mathbb{R}^d)$, the set of tempered distributions, \cite{vlad}, the following quantity is well defined:
\be
\overline{K}[e_\un]=\pin{K}{e_\un}.
\label{43}\en
What might exist, or not, is the following sum
\be
\pin{F}{G}_e=\sum_{\un}\overline{F}[e_\un]G[e_\un],
\label{44}\en
where $F(\ux), G(\ux)\in\Sc'(\mathbb{R}^d)$. This suggests the following:

\begin{defn}\label{def1}
	Two tempered distributions $F(\ux), G(\ux)\in\Sc'(\mathbb{R}^d)$ are $\F_e$-multiplicable if the series in (\ref{44}) converges.  
\end{defn}

What we have discussed before implies that all the square integrable functions are mutually $\F_e$-multiplicable, and  the result is independent of the specific choice of $\F_e$. This means that Definition \ref{def1} makes sense on a large set of tempered distributions, all those defined by ordinary square-integrable functions. Moreover, at least for these functions, (\ref{44}) and (\ref{41}) coincide. We will show in the next section that $\pin{F}{G}_e$ is also well defined in other cases. However, for generic elements of $\Sc'(\mathbb{R})$, it is not granted a priori that $\pin{F}{G}_e$ is independent of the choice of $\F_e$.

The following results are natural extensions of the properties of any {\em ordinary} scalar product to   $\pin{.}{.}_e$.

\vspace{2mm}

{\bf Result $\sharp 1$: } If $F(\ux), G(\ux)\in\Sc'(\mathbb{R}^d)$ are such that $\pin{F}{G}_e$ exists, then also $\pin{G}{F}_e$ exists and
\be
\pin{F}{G}_e=\overline{\pin{G}{F}_e}.
\label{45}\en
Indeed we have,  recalling that the series in (\ref{44}) converges, 
$$
\overline{\pin{F}{G}_e}=\sum_{\un}\overline{\pin{F}{e_\un}}\,\overline{\pin{e_\un}{G}}=\sum_{\un}\pin{G}{e_\un}\pin{e_\un}{F}=\pin{G}{F}_e,
$$
which in particular implies that $\pin{G}{F}_e$ exists, too.

\vspace{2mm}

{\bf Result $\sharp 2$: }  
If $F(\ux), G(\ux), L(\ux)\in\Sc'(\mathbb{R}^d)$ are such that $\pin{F}{G}_e$ and $\pin{F}{L}_e$ exist, then also $\pin{F}{\alpha G+\beta L}_e$ exists, for all $\alpha,\beta\in\mathbb{C}$, and
\be
\pin{F}{\alpha G+\beta L}_e=\alpha\,\pin{F}{ G}_e+\beta\pin{F}{L}_e.
\label{46}\en
Then the $\F_e$-multiplication is linear in the second variable. The proof is trivial and will not be given here. Of course, the $\F_e$-multiplication is anti-linear in the first variable.

\vspace{2mm}

{\bf Result $\sharp 3$: } 
If $F(\ux)\in\Sc'(\mathbb{R}^d)$ is such that $\pin{F}{F}_e$ exists, then $\pin{F}{F}_e\geq0$. In particular, if $\pin{F}{F}_e=0$, then $F[f]=0$ for all $f(x)\in\Lc_e$, the linear span of the $e_\un(\ux)$'s.

In fact, from (\ref{44}) we have
$$
\pin{F}{F}_e=\sum_{\un}|F[e_\un]|^2,
$$
which is never negative. Moreover, $\pin{F}{F}_e=0$ if and only if $F[e_\un]=0$ for all $\un$, which, because of the linearity of $F$, implies our claim.

\vspace{3mm}

{\bf Remark:--} It is not possible to conclude that $F=0$, even if $\pin{F}{F}_e=0$. The reason is that, to conclude that $F=0$, we should check that $F[g]=0$ for all $g(\ux)\in\Sc(\mathbb{R}^d)$. Now, since $\Sc(\mathbb{R}^d)\subset\ltwod$, it is clear that $g(\ux)=\|.\|-\lim_{\{N_k\},\infty}\sum_{\un=\underline 0}^{\underline N}\pin{e_\un}{g}\,e_\un(\ux)$, where $\underline N=(N_1,N_2,\ldots,N_d)$, with $N_j<\infty$, $\underline 0=(0,0,\ldots,0)$, and the convergence is in the norm of $\ltwod$. But this convergence does not imply that the same sequence converges in the topology $\tau_\Sc$ of $\Sc(\mathbb{R}^d)$. Hence, the continuity of $F$ is not sufficient to conclude that $$F[g]=\lim F\left[\sum_{\un=\underline 0}^{\underline N}\pin{e_\un}{g}\,e_\un(\ux)\right]=0.$$

\vspace{2mm}

It might be interesting to observe that for square-integrable functions $f(x)$ and $g(x)$ the adjoint of any bounded operator $X$ satisfies the equality $\pin{X^\dagger f}{g}_e=\pin{f}{Xg}_e$. This is a consequence of the analogous relation for $\pin{.}{.}$ and of the identity $\pin{f}{g}=\pin{f}{g}_e$, true $\forall f(\ux),g(\ux)\in\ltwod$. However, this is no longer granted for $F(\ux), G(\ux)\in\Sc'(\mathbb{R}^d)$. In fact, even if $F$ and $G$ are $\F_e$-multiplicable, and if $X^\dagger F, XG\in\Sc'(\mathbb{R}^d)$, there is no general reason for $\pin{X^\dagger F}{G}_e$ and $\pin{F}{XG}_e$ to exist, and to be equal. This is, in our opinion, one of the many points of the $\F_e$ multiplication which deserves a deeper investigation. Another relevant aspect of this multiplication concerns the {\em optimal} choice of $\F_e$, if any. We will comment on this particular aspect later on.

\section{Orthogonality of eigenstates}\label{sect5}

In Section \ref{sect3} we have deduced the vacua of the pseudo-bosonic lowering operators $A_j$ and $B_j^\dagger$. We will now use the standard pseudo-bosonic strategy, in its weak form, see \cite{bag2020JPA}-\cite{bagspringer}, to construct a set of distributions which are the (generalized) eigenstates of $H$ in (\ref{313}) and of its adjoint. In what follows we will use what we have discussed in Section \ref{sect4}, focusing on the case $d=2$.

In analogy with (\ref{A2}), after finding the vacua, the second step in our construction consists in using the raising pseudo-bosonic operators to construct, out of the vacua, two families of vectors. In particular, we put
\be
\varphi_{n_1,n_2}(x_1,x_2)=\frac{1}{\sqrt{n_1!\,n_2!}}B_1^{n_1}B_2^{n_2}\varphi_{0,0}(x_1,x_2),
\label{51}\en
and
\be
\psi_{n_1,n_2}(x_1,x_2)=\frac{1}{\sqrt{n_1!\,n_2!}}(A_1^\dagger)^{n_1}(A_2^\dagger)^{n_2}\psi_{0,0}(x_1,x_2),
\label{52}\en
where $n_1,n_2=0,1,2,3,\ldots$. As in Section \ref{sect4}, we will often use the notation $\un=(n_1,n_2)$ and $\ux=(x_1,x_2)$. These vectors are, clearly, not square-integrable functions. Indeed, they are tempered distributions, as it is already clear from their expressions for $\un=(0,0)$. Because of formulas (\ref{310bis}), (\ref{311}), and the fact that $p_k=-i\frac{\partial}{\partial x_k}$, $k=1,2$, it follows that $\varphi_{n_1,n_2}(x_1,x_2)$ and $\psi_{n_1,n_2}(x_1,x_2)$ are deduced from the vacua by acting on them with weak derivatives and multiplication operators, all operations mapping $\Sc'(\mathbb{R}^2)$ into itself. This implies that the two sets $\F_\varphi=\{\varphi_{\un}(\ux), \, n_1,n_2\geq0\}$ and $\F_\psi=\{\psi_{\un}(\ux), \, n_1,n_2\geq0\}$ are both sets of tempered distributions: $\varphi_{\un}(\ux),\psi_{\un}(\ux)\in\Sc'(\mathbb{R}^2)$, $\forall \un$. We would like to check now if, and in which sense, we can recover for these vectors the analogous of formula (\ref{A4}). In other words, we would like to understand if (and again, in which sense) $\F_\varphi$ and $\F_\psi$ are biorthonormal families of vectors.

\vspace{2mm}

{\bf Remark:--} It is not hard to check that (\ref{41}) can be used to define an {\em extended scalar product} between $\varphi_{\uzero}(\ux)$ and $\psi_{\uzero}(\ux)$, and to check that $\pin{\varphi_{\uzero}}{\psi_{\uzero}}=1$. However, this same definition does not allow any simple computation of the other scalar products $\pin{\varphi_{\un}}{\psi_{\um}}$, in general, and this is the main reason why we prefer to adopt here the definition of the $\F_e$-multiplication proposed in Section \ref{sect4}, see Definition \ref{def1}. This analysis has an important side effect, since it will allow us to check on a rather concrete situation that the definition of $\pin{.}{.}_e$ works also outside $\ltwod$. Hence, $\pin{.}{.}_e$ can really be  seen as an extension of the ordinary scal product in $\ltwod$.

\vspace{2mm}

Let $\G_e=\{e_n(x), \,n\geq0\}$ be the usual orthonormal basis of eigenfunctions of the harmonic oscillator, $$\tilde H_0=\frac{p^2}{2m_1}+\frac{1}{2}\,m_1\omega_1^2x^2,$$
where $m_1$ and $\omega_1$ are those introduced in Section \ref{sect3}. It is well known that $e_n(x)\in\Sc(\mathbb{R})$ for all $n\geq0$. Then we consider
\be
e_\un(\ux)=e_{n_1}(x_1)e_{n_2}(x_2),
\label{53}\en
and $\F_e=\{e_\un\}$. This is an orthonormal basis of $\ltwodue$ of functions, all belonging to $\scrdue$. Further, and very important, the various vectors of this basis obey ladder equalities when considered together with the bosonic operators $a_j$ and $a_j^\dagger$ in (\ref{310bis}). For instance, $$a_1^\dagger e_\un=\sqrt{n_1+1}\,e_{n_1+1,n_2},\qquad  a_2 e_\un=\sqrt{n_2}\,e_{n_1,n_2-1},$$ (or $ a_2 e_\un=0$ if $n_2=0$), and so on.

From now on we will use this particular basis to define $\pin{.}{.}_e$ as in (\ref{44}). The reason is simple: our raising operators $A_j^\dagger$ and $B_j$ can be written, see (\ref{311}), in terms of $a_j$ and $a_j^\dagger$, and their action on each $e_\un(\ux)$ in (\ref{53})  is simple\footnote{It is clear that for different physical system our choice of $\F_e$ could be different from the one here.}. Our main effort is to check that $\pin{\psi_{\uk}}{\varphi_{\ul}}_e$ makes sense, and that it is equal to $\delta_{\uk,\ul}$. This would imply that our multiplication in (\ref{44}) is useful, defined on a large set, and that the families  $\F_\varphi$ and $\F_\psi$ are biorthonormal (with respect to this extended scalar product).

To prove this claim, we need to compute
\be
\pin{\psi_{\uk}}{\varphi_{\ul}}_e=\sum_\un\overline{\psi_{\uk}}[e_\un]\,\varphi_{\ul}[e_\un],
\label{54}\en
using (\ref{44}). A general proof of the existence of this quantity is not easy. On the other hand, a direct check of the facts that $\pin{\psi_{\uk}}{\varphi_{\ul}}_e$ exists and that it is equal to $\delta_{\uk,\ul}$, it is not hard, if we restrict to few (low) values of $\uk$ and $\ul$. To do so it is convenient to check first the following equalities:
\be\left\{
\begin{array}{ll}
\varphi_{\uzero}[e_\un]=\alpha \delta_{n_1,n_2}, \qquad\qquad\qquad\qquad\qquad\varphi_{0,1}[e_\un]=\alpha \sqrt{2n_2}\delta_{n_1,n_2-1},\\ \varphi_{1,0}[e_\un]=\alpha \sqrt{2(n_2+1)}\delta_{n_1,n_2+1},\quad\qquad \varphi_{\uuno}[e_\un]=\alpha (1+2n_2)\delta_{n_1,n_2},\\
\end{array}
\right.
\label{55}\en
and
\be\left\{
\begin{array}{ll}
	\overline{\psi_{\uzero}}[e_\un]=\overline{\beta} (-1)^{n_2} \delta_{n_1,n_2}, \qquad\qquad\qquad\qquad\qquad\overline{\psi_{0,1}}[e_\un]=\overline{\beta} (-1)^{n_2+1}\sqrt{2n_2}\, \delta_{n_1,n_2-1},\\ \overline{\psi_{1,0}}[e_\un]=\overline{\beta} (-1)^{n_2}\sqrt{2(n_2+1)}\delta_{n_1,n_2+1},\quad\qquad \overline{\psi_{\uuno}}[e_\un]=\overline{\beta} (-1)^{n_2+1}(1+2n_2)\delta_{n_1,n_2}.\\
\end{array}
\right.
\label{56}\en
The proof of (some of) these identities is given in Appendix B.

We can now use (\ref{55}) and (\ref{56}) to check few of the orthonormality results needed to conclude that $\F_\varphi$ and $\F_\psi$ are $\F_e$-biorthonormal. In particular, it is easy to check that all the vectors $\psi_\uk$ and $\varphi_\ul$ are $\F_e$-orthogonal if $\uk\neq\ul$, and $\uk,\ul=(j_1,j_2)$, for all $j_1,j_2=0,1$. For instance, we have
$$
\pin{\psi_{1,0}}{\varphi_\uzero}_e=\sum_\un \overline{\psi_{1,0}}[e_\un]\varphi_{\uzero}[e_\un]=\overline{\beta}\alpha\sum_\un(-1)^{n_2}\sqrt{2(n_2+1)}\delta_{n_1,n_2+1}\delta_{n_1,n_2}=0,
$$
clearly. Similarly we have
$$
\pin{\psi_{\uuno}}{\varphi_{1,0}}_e=\sum_\un \overline{\psi_{\uuno}}[e_\un]\varphi_{1,0}[e_\un]=\overline{\beta}\alpha\sum_\un(-1)^{n_2+1}(1+2n_2)\delta_{n_1,n_2}\sqrt{2(n_2+1)}\delta_{n_1,n_2+1}=0,
$$
too, and so on. Much more interesting is the proof that, when $\uk=\ul$, $\pin{\psi_{\uk}}{\varphi_\ul}_e=1$, even in these few simple cases. Using the results in (\ref{55}) and (\ref{56}) we can get the following identities:
\be\left\{
\begin{array}{ll}
\pin{\psi_\uzero}{\varphi_\uzero}_e=\overline{\beta}\alpha\sum_k(-1)^k,\\
\pin{\psi_{1,0}}{\varphi_{1,0}}_e=\pin{\psi_{0,1}}{\varphi_{0,1}}_e=2\overline{\beta}\alpha\sum_k(-1)^k(k+1),\\
\pin{\psi_{\uuno}}{\varphi_{\uuno}}_e=-\overline{\beta}\alpha\sum_k(-1)^k(2k+1)^2.\\
\end{array}
\right.
\label{57}\en
None of these series converges in ordinary sense. However, they are all Abel-convergent. In fact, since
$$
A-\sum_k(-1)^k=\frac{1}{2}, \qquad A-\sum_k(-1)^kk=-\frac{1}{4}, \qquad  A-\sum_k(-1)^kk^2=0,
$$
if we take $\alpha\overline{\beta}=2$ we conclude that 
$$
\pin{\psi_\uzero}{\varphi_\uzero}_e=\pin{\psi_{1,0}}{\varphi_{1,0}}_e=\pin{\psi_{0,1}}{\varphi_{0,1}}_e=\pin{\psi_{\uuno}}{\varphi_{\uuno}}_e=1.
$$

Of course, what we have explicitly checked here is not a general result. In other words, this is just an indication that, see (\ref{54}), $\pin{\psi_{\uk}}{\varphi_{\ul}}_e$ exists and is equal to $\delta_{\uk,\ul}$. This is what we will discuss next.

\subsection{From few to many}

Our aim now is trying to generalize, as much as possible, formulas in (\ref{57}) so to check that $\F_\varphi$ and $\F_\psi$ are indeed $\F_e$-biorthonormal. However, as we will see, the existence of $\pin{\psi_{\uk}}{\varphi_{\ul}}_e$ will be, for us, a working assumption, motivated by the results we have deduced previously. We hope to be able to produce a general proof of this existence in a close future. It is quite likely that the difficulty in proving this result is connected with the fact that the various series above are not convergent in the usual sense, but only Abel-convergent. For this reason we believe that the preliminary analysis proposed here and before is relevant and useful for a deeper understanding of the situation.

We start proving that, calling $N_j=B_jA_j$ and $N_j^\dagger$ its adjoint, the following weak eigenvalue equations are satisfied:
\be
\pin{\Phi}{N_j\varphi_\ul}=l_j\pin{\Phi}{\varphi_\ul}, \qquad \pin{\Phi}{N_j^\dagger\psi_\ul}=l_j\pin{\Phi}{\psi_\ul},
\label{58}\en
for all $\ul\in\mathbb{N}_0^2$, for all $\Phi(x_1,x_2)\in\scrdue$, and for $j=1,2$.

Using the definition of $N_1=B_1A_1$ in terms of multiplication and derivative operators, see (\ref{310bis}) and (\ref{311}), and using the fact that $\Phi(x_1,x_2)\in\scrdue$ and that $\scrdue$ is stable under the action of the various operators involved in the game, we have
$$
\pin{\Phi}{N_1\varphi_\un}=\pin{N_1^\dagger\Phi}{\varphi_\un}=\frac{1}{\sqrt{n_1!}}\pin{N_1^\dagger\Phi}{B_1^{n_1}\varphi_{0,n_2}}=\frac{1}{\sqrt{n_1!}}\pin{{B_1^\dagger}^{n_1}N_1^\dagger\Phi}{\varphi_{0,n_2}}.
$$
Now, since $\Phi(x_1,x_2)\in\scrdue$, which is stable, we can safely rewrite
$$
{B_1^\dagger}^{n_1}N_1^\dagger\Phi=\left(\left[{B_1^\dagger}^{n_1},N_1^\dagger\right]+N_1^\dagger{B_1^\dagger}^{n_1}\right)\Phi=\left(n_1{B_1^\dagger}^{n_1}+N_1^\dagger{B_1^\dagger}^{n_1}\right)\Phi,
$$
where the last equality can be proved by induction on $n_1$. Hence, using again the definition of the weak derivative, we have
$$
\frac{1}{\sqrt{n_1!}}\pin{{B_1^\dagger}^{n_1}N_1^\dagger\Phi}{\varphi_{0,n_2}}=\frac{1}{\sqrt{n_1!}}\pin{(n_1{B_1^\dagger}^{n_1}+N_1^\dagger{B_1^\dagger}^{n_1})\Phi}{\varphi_{0,n_2}}=n_1\pin{\Phi}{\varphi_\un},
$$
since, in particular, $\pin{N_1^\dagger{B_1^\dagger}^{n_1}\Phi}{\varphi_{0,n_2}}=\pin{{B_1^\dagger}^{n_1}\Phi}{N_1\varphi_{0,n_2}}=0$. The other equalities in (\ref{58}) can be proved in a similar way.

This result has interesting consequences like those given in the rest of this section.

\begin{prop}\label{prop2}
	Assume that, for some $\uk$ and $\ul$, $\pin{\psi_\uk}{\varphi_\ul}_e$ exists. Then $\pin{\psi_\uk}{N_j\varphi_\ul}_e$ and $\pin{N_j^\dagger\psi_\uk}{\varphi_\ul}_e$, $j=1,2$, also exist and
	\be
	\pin{\psi_\uk}{N_j\varphi_\ul}_e=l_j\pin{\psi_\uk}{\varphi_\ul}_e, \qquad \pin{N_j^\dagger\psi_\uk}{\varphi_\ul}_e=k_j\pin{\psi_\uk}{\varphi_\ul}_e.
	\label{510}\en
\end{prop} 

\begin{proof}
	By definition we have, for instance,
	$$
	\pin{\psi_\uk}{N_1\varphi_\ul}_e=\sum_\un\overline{\psi_{\uk}}[e_\un]\,(N_1\varphi_{\ul})[e_\un].
	$$
	But, since $e_\un(\ux)\in\scrdue$, we can use (\ref{58}) and we get $(N_1\varphi_{\ul})[e_\un]=\pin{e_\un}{N_1\varphi_{\ul}}=l_1\pin{e_\un}{\varphi_{\ul}}$. Hence 
	$$
		\pin{\psi_\uk}{N_j\varphi_\ul}_e=l_1\sum_\un\overline{\psi_{\uk}}[e_\un]\,\varphi_{\ul}[e_\un]=l_1\pin{\psi_\uk}{\varphi_\ul}_e,
	$$
	as we had to prove. The other equalities in (\ref{510}) can be proved in a similar way.

\end{proof}

We have already commented that, in general, taken $F$ and $G$ in $\Sc'(\mathbb{R}^2)$, and a given operator $X$, even if $X^\dagger F, XG\in\Sc'(\mathbb{R}^2)$, there is no general reason for $\pin{X^\dagger F}{G}_e$ and $\pin{F}{XG}_e$ to exist, and to have $\pin{X^\dagger F}{G}_e=\pin{F}{XG}_e$. However, this may happen in some particular cases. In fact, we can check that
\be
\pin{N_j^\dagger\psi_\uk}{\varphi_\ul}_e=\pin{\psi_\uk}{N_j\varphi_\ul}_e
\label{511}\en
$\forall \uk,\ul$, $j=1,2$. The proof of this equality is rather technical and is given in Appendix C. Now it is clear that, as in the standard, Hilbert space, settings, the following result holds:
\begin{prop}\label{prop3}
	If $\uk\neq\ul$ then
	\be
	\pin{\psi_\uk}{\varphi_\ul}_e=0.
	\label{512}\en
\end{prop}

\begin{proof}
	The proof is identical to the usual one, using the (non trivial, here) identities in (\ref{510}) and (\ref{511}).

\end{proof}

Summarizing, what we get here is a strong indication that the results deduced before, using (\ref{55}) and (\ref{56}), can be generalized and that the sets $\F_\varphi$ and $\F_\psi$ are indeed two $\F_e$-biorthonormal sets of tempered distributions, and are made of weak eigenstates of the number operators $N_j$ and $N_j^\dagger$ and, therefore, of weak eigenstates of the Hamiltonian $H$ in (\ref{313}) and of $H^\dagger$. In our analysis, we were somehow forced to introduce a new, potentially interesting, multiplication between distributions.

\section{Conclusions}\label{sect6}

There are several open points in the analysis proposed in this paper. First of all, we believe that the idea of $\F_e$-multiplication of distributions may have interesting consequences and applications. Also, their properties should be analyzed in more details than those considered here. We hope to consider this aspect of our analysis soon.

There are also other aspects of the system considered here which should be clarified: if on one side we were able to compute explicitly some of the scalar products $\pin{\psi_\uk}{\varphi_\ul}_e$, showing that, as expected, they are zero or one, we have no general argument showing that $\pin{\psi_\uk}{\varphi_\ul}_e$ does exist for all $\uk$ and $\ul$. Still, under the hypothesis that this exists, we were able to deduce several interesting results, including the biorthogonality of the vectors $\psi_\uk$ and $\varphi_\ul$. However, the nature of $\F_\varphi$ and $\F_\psi$ as possible bases (of some kind, see \cite{baginbagbook}) is also still to be understood.

Another interesting aspect is the role of the Abel summation appearing in several computations. This creates more questions: is this related to the physical system? How?

Finally, on a more physical side, the Lagrangian considered in (\ref{31}) and (\ref{32}) is just a particular possible choice among many possibilities. What does it change if we consider a different $L_1$? Which kind of physics can we describe? And what is the role of distributions, if any, for general dissipative systems?

We hope to be able to answer some of these questions soon.

\section*{Acknowledgements}

The author acknowledges partial financial support from Palermo University (via FFR2021 "Bagarello") and from G.N.F.M. of the INdAM. The author is grateful to Vincenzo Sottile, Gioele Arcoleo, Federico Roccati and Francesco Gargano for many discussions at a preliminary stage of this paper.

\section*{Data accessibility statement}

This work does not have any experimental data.

%\section*{Competing interests statement}
%
%We have no competing interests.
%
%\section*{Authors' contributions}
%
%FB proposed the model and its solution. AK and EH worked on the
%interpretation of the model in the context of DM.
%
%\section*{Acknowledgements}
%
%One of the authors (FB) acknowledges partial support from the University of
%Palermo and from G.N.F.M. and the discussion on the paper was started during
%his visit to Linnaeus University (supported by the  grant ``Mathematical
%Modeling of Complex Hierarchic Systems''). FB also wishes to thank Prof. A.
%Busacca for his strong support for this project.

\section*{Funding statement}

The author acknowledges partial financial support from Palermo University, via FFR2021 "Bagarello".

\renewcommand{\theequation}{A.\arabic{equation}}

\section*{Appendix A: a micro review on pseudo-bosons}\label{appendixA}

Let $\Hil$ be a given Hilbert space with scalar product $\left<.,.\right>$ and related norm $\|.\|$. 

\vspace{2mm}

Let $a$ and $b$ be two operators
on $\Hil$, with domains $D(a)$ and $D(b)$ respectively, $a^\dagger$ and $b^\dagger$ their adjoint, and let $\D$ be a dense subspace of $\Hil$
such that $a^\sharp\D\subseteq\D$ and $b^\sharp\D\subseteq\D$, where with $x^\sharp$ we indicate $x$ or $x^\dagger$. Of course, $\D\subseteq D(a^\sharp)$
and $\D\subseteq D(b^\sharp)$.

\begin{defn}\label{defc21}
	The operators $(a,b)$ are $\D$-pseudo bosonic  if, for all $f\in\D$, we have
	\be
	a\,b\,f-b\,a\,f=f.
	\label{C1}\en
\end{defn}

When CCR are replaced by (\ref{C1}), it is necessary to impose some reasonable conditions which are verified in explicit models. In particular, our starting assumptions are the following:

\vspace{2mm}

{\bf Assumption $\D$-pb 1.--}  there exists a non-zero $\varphi_{ 0}\in\D$ such that $a\,\varphi_{ 0}=0$.

\vspace{1mm}

{\bf Assumption $\D$-pb 2.--}  there exists a non-zero $\Psi_{ 0}\in\D$ such that $b^\dagger\,\Psi_{ 0}=0$.

\vspace{2mm}

It is obvious that, since $\D$ is stable under the action of the operators introduced above,  $\varphi_0\in D^\infty(b):=\cap_{k\geq0}D(b^k)$ and  $\Psi_0\in D^\infty(a^\dagger)$, so
that the vectors \be \varphi_n:=\frac{1}{\sqrt{n!}}\,b^n\varphi_0,\qquad \Psi_n:=\frac{1}{\sqrt{n!}}\,{a^\dagger}^n\Psi_0, \label{A2}\en
$n\geq0$, can be defined and they all belong to $\D$. Then, they also belong to the domains of $a^\sharp$, $b^\sharp$ and $N^\sharp$, where $N=ba$. We see that, from a practical point of view, $\D$ is the natural space to work with and, in this sense, it is even more relevant than $\Hil$. Let's put $\F_\Psi=\{\Psi_{ n}, \,n\geq0\}$ and
$\F_\varphi=\{\varphi_{ n}, \,n\geq0\}$.
It is  simple to deduce the following lowering and raising relations:
\be
\left\{
\begin{array}{ll}
	b\,\varphi_n=\sqrt{n+1}\varphi_{n+1}, \qquad\qquad\quad\,\, n\geq 0,\\
	a\,\varphi_0=0,\quad a\varphi_n=\sqrt{n}\,\varphi_{n-1}, \qquad\,\, n\geq 1,\\
	a^\dagger\Psi_n=\sqrt{n+1}\Psi_{n+1}, \qquad\qquad\quad\, n\geq 0,\\
	b^\dagger\Psi_0=0,\quad b^\dagger\Psi_n=\sqrt{n}\,\Psi_{n-1}, \qquad n\geq 1,\\
\end{array}
\right.
\label{C3}\en as well as the eigenvalue equations $N\varphi_n=n\varphi_n$ and  $N^\dagger\Psi_n=n\Psi_n$, $n\geq0$. In particular, as a consequence
of these last two equations,  if we choose the normalization of $\varphi_0$ and $\Psi_0$ in such a way $\left<\varphi_0,\Psi_0\right>=1$, we deduce that
\be \left<\varphi_n,\Psi_m\right>=\delta_{n,m}, \label{A4}\en
for all $n, m\geq0$. Hence $\F_\Psi$ and $\F_\varphi$ are biorthogonal.

The analogy with ordinary bosons suggests us to consider the following:

\vspace{2mm}

{\bf Assumption $\D$-pb 3.--}  $\F_\varphi$ is a basis for $\Hil$.

\vspace{1mm}

This is equivalent to requiring that $\F_\Psi$ is a basis for $\Hil$ as well. However, several  physical models show that $\F_\varphi$ is {\bf not} a basis for $\Hil$, but it is still complete in $\Hil$. This suggests to adopt the following weaker version of  Assumption $\D$-pb 3, \cite{baginbagbook}:

\vspace{2mm}

{\bf Assumption $\D$-pbw 3.--}  For some subspace $\G$ dense in $\Hil$, $\F_\varphi$ and $\F_\Psi$ are $\G$-quasi bases.

\vspace{2mm}
This means that, for all $f$ and $g$ in $\G$,
\be
\left<f,g\right>=\sum_{n\geq0}\left<f,\varphi_n\right>\left<\Psi_n,g\right>=\sum_{n\geq0}\left<f,\Psi_n\right>\left<\varphi_n,g\right>,
\label{A4b}
\en
which can be seen as a weak form of the resolution of the identity, restricted to $\G$.

We refer to \cite{bagspringer,baginbagbook} for more details.

\renewcommand{\theequation}{A.\arabic{equation}}

\section*{Appendix B: on formulas (\ref{55}) and (\ref{56})}\label{appendixB}

Let us compute first $\varphi_{\uzero}[e_\un]$. Since  $\varphi_{00}(x_1,x_2)=\alpha\delta(x_1-x_2)$, from (\ref{53}) we get
$$
\varphi_{\uzero}[e_\un]=\alpha \int_{\mathbb{R}^2}dx_1\,dx_2\delta(x_1-x_2)e_{n_1}(x_1)e_{n_2}(x_2)=\alpha \int_{\mathbb{R}}dx_1e_{n_1}(x_1)e_{n_2}(x_1)=\alpha\delta_{n_1,n_2},
$$ 
due to the orthonormality of the $e_n(x)$'s. Similarly, recalling that $e_n(-x)=(-1)^ne_n(x)$, we have
$$
\overline{\psi_{\uzero}}[e_\un]=\overline{\beta} \int_{\mathbb{R}^2}dx_1\,dx_2\delta(x_1+x_2)e_{n_1}(x_1)e_{n_2}(x_2)=\overline{\beta} \int_{\mathbb{R}}dx_1\,e_{n_1}(x_1)e_{n_2}(-x_1)=
$$ 
$$
=\overline{\beta}(-1)^{n_2} \int_{\mathbb{R}}dx_1\,e_{n_1}(x_1)e_{n_2}(x_1)=\overline{\beta}(-1)^{n_2}\delta_{n_1,n_2}.
$$
Next we compute $\varphi_{0,1}[e_\un]=\pin{e_\un}{\varphi_{0,1}}=\pin{e_\un}{B_2\varphi_{\uzero}}=\pin{B_2^\dagger e_\un}{\varphi_{\uzero}}$, where the last equality follows from the definition of the weak derivative of distributions. Now,  
$$B_2^\dagger e_\un=\frac{1}{\sqrt{2}}(a_1^\dagger+a_2)e_\un=\frac{1}{\sqrt{2}}\left(\sqrt{n_1+1}e_{n_1+1,n_2}+\sqrt{n_2}e_{n_1,n_2-1}\right),
$$
 the last term being zero if $n_2=0$. Then, in view of what deduced before, we have $\varphi_{\uzero}[e_{n_1+1,n_2}]=\alpha \delta_{n_1+1,n_2}$ and $\varphi_{\uzero}[e_{n_1,n_2-1}]=\alpha \delta_{n_1,n_2-1}$, so that $\varphi_{0,1}[e_\un]=\alpha \sqrt{2n_2}\delta_{n_1,n_2-1}$ follows. We conclude this Appendix by proving the identity $\varphi_{\uuno}[e_\un]=\alpha (1+2n_2)\delta_{n_1,n_2}$.  Similar computations can be repeated to check all the other identities in (\ref{55}) and (\ref{56}).

We have
$$
\varphi_{\uuno}[e_\un]=\pin{e_\un}{\varphi_\uuno}=\pin{e_\un}{B_1B_2\varphi_\uzero}=\pin{B_2^\dagger B_1^\dagger e_\un}{\varphi_\uzero}=
$$
$$
=\frac{1}{2}\pin{\left(a_1^\dagger a_1+a_1^\dagger a_2^\dagger+a_1a_2+a_2a_2^\dagger\right)e_\un}{\varphi_\uzero}=$$
$$=\frac{1}{2}\left(n_1\varphi_{\uzero}[e_\un]+\sqrt{(n_1+1)(n_2+1)}\varphi_{\uzero}[e_\un+\uuno]+\sqrt{n_1n_2}\varphi_{\uzero}[e_\un-\uuno]+(1+n_2)\varphi_{\uzero}[e_\un]\right)=
$$
$$
=\frac{\alpha}{2}\delta_{n_1,n_2}\left(n_1+\sqrt{(n_1+1)(n_2+1)}+\sqrt{n_1n_2}+(1+n_2)\right),
$$
from which our result in (\ref{55}) follows. Notice that we have used here the notation $\un\pm\uuno=(n_1\pm1,n_2\pm2)$, and used three times the identity $\varphi_{\uzero}[e_\un]=\alpha\delta_{n_1,n_2}$ proved before.

\renewcommand{\theequation}{A.\arabic{equation}}

\section*{Appendix C: on the identity in (\ref{511})}\label{appendixC}

We start rewriting
\be
\pin{\psi_\uk}{N_1\varphi_\ul}_e=\sum_\un\overline{\psi_{\uk}}[e_\un]\,(N_1\varphi_{\ul})[e_\un].
\label{c1}\en
We have already seen in the Proposition \ref{prop3} that $(N_1\varphi_{\ul})[e_\un]=l_1\varphi_{\ul}[e_\un]$. Here we need to rewrite it in a different form, noticing first that $(N_1\varphi_{\ul})[e_\un]=\varphi_{\ul}[N_1^\dagger e_\un]$. Next, using (\ref{311}), is possible to compute $N_1^\dagger e_\un$ and to check that
$$
\varphi_{\ul}[N_1^\dagger e_\un]=\frac{1}{2}\left((n_1-n_2-1)\varphi_{\ul}[e_\un]+\sqrt{(n_1+1)(n_2+1)}\,\varphi_{\ul}[e_{\un+\uuno}]-\sqrt{n_1n_2}\,\varphi_{\ul}[e_{\un-\uuno}]\right).
$$
which, using (\ref{c1}), returns
$$
\pin{\psi_\uk}{N_1\varphi_\ul}_e=$$
$$=\frac{1}{2}\sum_\un\overline{\psi_{\uk}}[e_\un]\left((n_1-n_2-1)\varphi_{\ul}[e_\un]+\sqrt{(n_1+1)(n_2+1)}\,\varphi_{\ul}[e_{\un+\uuno}]-\sqrt{n_1n_2}\,\varphi_{\ul}[e_{\un-\uuno}]\right).
$$ It is clear that the series converge since the left-hand side in (\ref{c1}) exists, because of  (\ref{510}) and of our working assumptions.
In a similar way we have
$$
\pin{N_1^\dagger\psi_\uk}{\varphi_\ul}_e=\sum_\un\overline{\psi_{\uk}}[N_1e_\un]\,\varphi_{\ul}[e_\un],
$$
and computing $N_1e_\un$, we get
$$
\overline{\psi_{\uk}}[N_1e_\un]=\frac{1}{2}\left((n_1-n_2-1)\overline{\psi_{\uk}}[e_\un]-\sqrt{(n_1+1)(n_2+1)}\,\overline{\psi_{\uk}}[e_{\un+\uuno}]+\sqrt{n_1n_2}\,\overline{\psi_{\uk}}[e_{\un-\uuno}]\right),
$$
so that 
$$
\pin{N_1^\dagger\psi_\uk}{\varphi_\ul}_e=
$$
$$
=\frac{1}{2}\sum_\un\left((n_1-n_2-1)\overline{\psi_{\uk}}[e_\un]-\sqrt{(n_1+1)(n_2+1)}\,\overline{\psi_{\uk}}[e_{\un+\uuno}]+\sqrt{n_1n_2}\,\overline{\psi_{\uk}}[e_{\un-\uuno}]\right)\varphi_{\ul}[e_\un].
$$
Now the fact that $\pin{\psi_\uk}{N_1\varphi_\ul}_e=\pin{N_1^\dagger\psi_\uk}{\varphi_\ul}_e$ follows from a direct comparisons of the two results deduced here (with some change of variable).

The others identities can be proved in a similar way.

\end{document}